\documentclass[twoside,twocolumn,english,aps,prl,superscriptaddress,showpacs]{revtex4-1}
\usepackage[T1]{fontenc}
\usepackage[utf8]{inputenc}
\usepackage{xcolor}
\usepackage{pdfcolmk}
\usepackage{babel}
\usepackage{amsmath}
\usepackage{amsthm}
\usepackage{amssymb}
\PassOptionsToPackage{normalem}{ulem}
\usepackage{ulem}
\usepackage[unicode=true,pdfusetitle,
 bookmarks=true,bookmarksnumbered=false,bookmarksopen=false,
 breaklinks=false,pdfborder={0 0 1},backref=false,colorlinks=false]
 {hyperref}

\makeatletter

\providecolor{lyxadded}{rgb}{0,0,1}
\providecolor{lyxdeleted}{rgb}{1,0,0}

\DeclareRobustCommand{\lyxsout}[1]{\ifx\\#1\else\sout{#1}\fi}

\theoremstyle{plain}
\newtheorem{thm}{\protect\theoremname}
\theoremstyle{definition}
\newtheorem{defn}{\protect\definitionname}
\theoremstyle{plain}
\newtheorem{prop}{\protect\propositionname}
\theoremstyle{plain}
\newtheorem*{cor*}{\protect\corollaryname}
\theoremstyle{plain}
\newtheorem{lem}{\protect\lemmaname}

\usepackage{times}
\DeclareRobustCommand\textprime{\leavevmode \raise.8ex\hbox{\text@char\scriptfont\prime}}

\DeclareMathOperator{\tr}{tr}

\newcommand{\1}{\leavevmode{\rm 1\ifmmode\mkern  -4.8mu\else\kern -.3em\fi I}}

\newtheorem{innercustomgeneric}{\customgenericname}
\providecommand{\customgenericname}{}
\newcommand{\newcustomtheorem}[2]{%
  \newenvironment{#1}[1]
  {%
   \renewcommand\customgenericname{#2}%
   \renewcommand\theinnercustomgeneric{##1}%
   \innercustomgeneric
  }
  {\endinnercustomgeneric}
}

\newcustomtheorem{customthm}{Theorem}
\newcustomtheorem{customprop}{Proposition}
\newcustomtheorem{customdef}{Definition}

\makeatletter

\def\thm@space@setup{\thm@preskip=0pt
\thm@postskip=0pt}

\def\bbl@set@language#1{%
  \edef\languagename{%
    \ifnum\escapechar=\expandafter`\string#1\@empty
    \else\string#1\@empty\fi}%
  \@ifundefined{babel@language@alias@\languagename}{}{%
    \edef\languagename{\@nameuse{babel@language@alias@\languagename}}%
  }%
  \select@language{\languagename}%
  \expandafter\ifx\csname date\languagename\endcsname\relax\else
    \if@filesw
      \protected@write\@auxout{}{\string\select@language{\languagename}}%
      \bbl@for\bbl@tempa\BabelContentsFiles{%
        \addtocontents{\bbl@tempa}{\xstring\select@language{\languagename}}}%
      \bbl@usehooks{write}{}%
    \fi
  \fi}
\newcommand{\DeclareLanguageAlias}[2]{%
  \global\@namedef{babel@language@alias@#1}{#2}%
}
\makeatother

\DeclareLanguageAlias{en}{english}

\makeatother

\providecommand{\corollaryname}{Corollary}
\providecommand{\definitionname}{Definition}
\providecommand{\lemmaname}{Lemma}
\providecommand{\propositionname}{Proposition}
\providecommand{\theoremname}{Theorem}

\usepackage{newtxtext,newtxmath}

\begin{document}
\title{Ergodicity, eigenstate thermalization, and the foundations of statistical
mechanics \\
in quantum and classical systems}
\author{Lorenzo Campos Venuti}
\affiliation{Department of Physics and Astronomy, University of Southern California,
CA 90089, USA}
\author{Lawrence Liu}
\affiliation{Department of Physics and Astronomy, University of Southern California,
CA 90089, USA}
\date{\today}
\begin{abstract}
Boltzmann's ergodic hypothesis furnishes a possible explanation for
the emergence of statistical mechanics in the framework of classical
physics. In quantum mechanics, the Eigenstate Thermalization Hypothesis
(ETH) is instead generally considered as a possible route to thermalization.
This is because the notion of ergodicity itself is vague in the quantum
world and it is often simply taken as a synonym for thermalization.
Here we show, in an elementary way, that when quantum ergodicity is
properly defined, it is, in fact, equivalent to ETH. In turn, ergodicity
is equivalent to thermalization, thus implying the equivalence of
thermalization and ETH. This result previously appeared in {[}De Palma
\emph{et al.}, Phys.~Rev.~Lett.~\textbf{115}, 220401 (2015){]},
but becomes particularly clear in the present context. We also show
that it is possible to define a classical analogue of ETH which is
implicitly assumed to be satisfied when constructing classical statistical
mechanics. Classical and quantum statistical mechanics are built according
to the familiar standard prescription. This prescription, however,
is ontologically justified only in the quantum world.
\end{abstract}
\maketitle

\paragraph*{Introduction. }

A possible mechanistic justification of classical statistical mechanics
proceeds via the ergodic hypothesis of Boltzmann, i.e., the assumption---to
be proven in the cases at hand---that a given classical Hamiltonian
dynamical system is ergodic. As we will see, this is, in fact, not
the whole story and more is needed. In any case, in the quantum world
it is not entirely clear what constitutes a meaningful notion of ergodicity,
let alone whether or not such a notion implies thermalization as it
does classically. In fact often quantum ergodicity is not separately
defined but simply taken as a synonym for thermalization (see e.g.~\citep{pal_many-body_2010}
or footnote 1 in \citep{abanin_ergodicity_2018}). Recently, the Eigenstate
Thermalization Hypothesis (ETH) has emerged as a promising hypothesis
to explain thermalization in the framework of quantum mechanics as
ETH trivially implies thermalization \citep{deutsch_quantum_1991,srednicki_chaos_1994,rigol_relaxation_2007,gogolin_equilibration_2016}\footnote{Additionally, in the mathematical literature, the related concept
of \emph{quantum (unique) ergodicity}, which we might call ETH at
infinite energy, has been proved rigorously for single particle Hamiltonians
whose classical counterpart is ergodic (see e.g.~\citep{zelditch_quantum_1990,sunada_quantum_1997}). }. Quantum ergodicity, however, can and has been precisely defined
in different settings (see e.g., \citep{mazur_non-ergodicity_1969,suzuki_ergodicity_1971,ruelle_statistical_1999,alicki_quantum_2001}).
In this paper we give a few characterizations of the notion of ergodicity
in the quantum world. It is shown that ergodicity is indeed equivalent
to thermalization. Moreover, ergodicity is also seen to be equivalent
to ETH, thus implying at once an equivalence between ETH and thermalization.
That ETH is in fact not only sufficient but also necessary for thermalization
was first proven in \citep{de_palma_necessity_2015} but becomes particularly
clear in our setting.

\paragraph*{Ergodicity in quantum physics.}

Our starting point is to give a meaningful definition of ergodicity
in the quantum setting. It will be useful first to recall various
equivalent characterizations of ergodicity in the classical setting
(see \citep{cornfeld_ergodic_1982}).
\begin{thm}
\label{Characterisations-of-ergodicity}(Characterizations of ergodicity)
Let $(M,g^{t},\mu)$ be a measure-preserving system. $M$ is a measure
space, $g^{t}$ a flow, and $\mu$ a normalized, $g^{t}$-invariant
measure on $M$. Denoting $\langle f\rangle_{\mu}\equiv\int_{M}\!f(x)\,\mathrm{d}\mu(x)$,
the following are equivalent:
\begin{enumerate}
\item Any (Borel) set $X\subseteq M$ which is almost invariant ($g^{t}(X)$
differs from $X$ by a null set for all $t$) has either full measure
or zero measure.
\item For any $f,g\in L^{\infty}(M,\mu)$, 
\[
\lim_{T\to\infty}\frac{1}{T}\int_{0}^{T}\!\mathrm{d}t\,\langle f(t)g\rangle_{\mu}=\langle f\rangle_{\mu}\langle g\rangle_{\mu}.
\]
\item For any $f\in L^{1}(M,\mu)$, the averages $T^{-1}\int_{0}^{T}\!\mathrm{d}t\,f\circ g^{t}$
converge pointwise almost everywhere to $\langle f\rangle_{\mu}$.
\end{enumerate}
\end{thm}
Let us now switch to quantum mechanics. The analogue of the triple
$(M,g^{t},\mu)$ (also valid in infinite dimension), is, not surprisingly,
given by a quantum dynamical system comprising a $C^{\ast}$-algebra,
a dynamical evolution, and a quantum state. In this paper we use the
standard point of view that the approach to thermodynamic equilibrium
can be understood by studying increasingly large systems of finite
size. Hence we make the assumption that the total Hilbert space $\mathcal{H}$
is finite-dimensional. The system's Hamiltonian acting on $\mathcal{H}$
can then be written as $H=\sum_{n}E_{n}\Pi_{n}$ ($E_{n}$ eigenenergies,
$\Pi_{n}$ possibly degenerate eigen-projectors). Since we are dealing
with an isolated system we consider, as usual, a collection $V$ of
energy eigenvalues (usually called a shell, and typically but not
necessarily of the form $V=\{E|\,\overline{E}\le E_{n}\le\overline{E}+\Delta\}$).
Let $\mathcal{H}_{V}$ be the corresponding Hilbert space, $\Pi_{V}$
the orthogonal projector onto it, $\Pi_{V}=\sum_{E_{n}\in V}\Pi_{n}$,
and $\mathcal{S}_{V}$ the set of quantum states with support on $\mathcal{H}_{V}$.
The Schr\"{o}dinger dynamics is $\mathcal{E}_{t}(\cdot)=U_{t}\cdot U_{t}^{\dagger}$
with $U_{t}=\mathrm{e}^{-\mathrm{i}tH}.$ The role played by the measure
$\mu$ is now taken by a quantum state $\rho_{V}\in\mathcal{S}_{V}$
invariant under the dynamics (whence $[\rho_{V},\Pi_{n}]=0$
for all $n$ and $\rho_{V}=\Pi_{V}\rho_{V}=\rho_{V}\Pi_{V}$). The
equivalent of the ``phase space average'' is $\langle A\rangle _{V}=\tr(A\rho_{V})$
where $A$ is an observable. In principle $A$ is defined only on
$\mathcal{H}_{V}$ but it is useful to consider observables defined
on the whole space: $A\in\mathcal{B}(\mathcal{H})$. Henceforth we
write $\overline{X(t)}=\lim_{T\to\infty}T^{-1}\int_{0}^{T}\!\mathrm{d}t\,X(t)$.

We now define thermalization for a specific observable. This definition
is essentially the same as in the classical case but we single out
a particular observable to leave open the possibility that only some
(but not all) observables thermalize. Moreover, for clarity of exposition
and in analogy with the classical case, we first consider exact ergodicity.
We will relax this condition later. 
\begin{defn}
We say that $A$ thermalizes on $\mathcal{H}_{V}$ (with equilibrium
state $\rho_{V}$) if $\overline{\tr(A(t)\rho_{0})}=\langle A\rangle_{V}$
for all $\rho_{0}\in\mathcal{S}_{V}$. 
\end{defn}
As we have seen there are several equivalent characterization of ergodicity
in classical dynamical systems. We first consider characterization
2 of Theorem \ref{Characterisations-of-ergodicity} which can be trivially
reformulated quantum mechanically. Once again we retain the possibility
of ergodicity only for some specific observables.
\begin{defn}
\label{V-ergodic}We say that an observable $A$ is ergodic on the
energy shell $V$ (shell-ergodic) if $\overline{\langle A(t)A\rangle_{V}}=(\langle A\rangle_{V})^{2}$. 
\end{defn}
This definition appears in \citep{mazur_non-ergodicity_1969,suzuki_ergodicity_1971}.
We now give an alternative characterization of ergodicity which may
help clarify its meaning. 
\begin{prop}
\label{Ergo-alt} The observable $A$ is shell-ergodic if and only
if $\overline{A(t)}\Pi_{V}=\langle A\rangle_{V}\Pi_{V}$.
\end{prop}
\begin{proof}
Note that $[\overline{A(t)},\Pi_{V}]=0$ so $\overline{A(t)}\Pi_{V}=\Pi_{V}\overline{A(t)}=\Pi_{V}\overline{A(t)}\Pi_{V}$.
The $\Leftarrow$ direction is clear: multiply by $A\rho_{V}$---since
$\rho_{V}\in\mathcal{S}_{V}$, $\rho_{V}=\Pi_{V}\rho_{V}\Pi_{V}$---and
take the trace. For the other direction, we use the auxiliary result
$\langle\overline{A(t)}\rangle_{V}=\langle A\rangle_{V}$ (also valid
classically). We define the dephasing operator $\mathcal{D}(A)\equiv\sum_{n}\Pi_{n}A\Pi_{n}=\overline{A(t)}$
and its complement $\mathcal{Q}=\1-\mathcal{D}$. Both $\mathcal{D}$
and $\mathcal{Q}$ are orthogonal projectors with respect to the Hilbert--Schmidt
scalar product $\langle X,Y\rangle_{HS}\equiv\tr(X^{\dagger}Y)$.
Since $\rho_{V}$ is diagonal in the energy eigenbasis $\rho_{V}=\mathcal{D}(\rho_{V})$.
So $\tr(A\rho_{V})=\langle A,\mathcal{D}\rho_{V}\rangle_{HS}=\langle\mathcal{D}A,\rho_{V}\rangle_{HS}=\langle\overline{A(t)}\rangle_{V}$.
In a similar way $\langle\overline{A(t)}A\rangle_{V}=\tr(A\rho_{V}\overline{A(t)})=\tr[A\mathcal{D}(\rho_{V}\overline{A(t)})]=\langle\overline{A(t)}\,\overline{A(t)}\rangle_{V}$.
With these results $\overline{\langle A(t)A\rangle_{V}}=(\langle A\rangle_{V})^{2}$
can be written as $\langle(\overline{A(t)}-\langle A\rangle_{V})^{2}\rangle_{V}=0$.
Finally, by Lemma \ref{lem:zero} (see Appendix A), $(\overline{A(t)}-\langle A\rangle_{V})\Pi_{V}=0$;
that is, $A$ is shell-ergodic.
\end{proof}
It should be clear that this is the quantum mechanical equivalent
to the standard ergodicity statement (characterization 3), according
to which time averages of functions are the constant functions a.e.~with
values given by the equilibrium averages. Indeed, the analogue of
a constant function in quantum mechanics is a projector. Moreover,
since invariant spaces in quantum mechanics are linear subspaces,
the usual restriction ``for almost any initial state'' loses its
meaning. We will see later how a similar condition can be re-introduced
in quantum mechanics. Note that this is essentially the definition
of ergodicity given for abstract $C^{*}$-algebras \citep{ruelle_statistical_1999,alicki_quantum_2001}.
There, however, the statement is taken for all observables in the
algebra and here the shell Hilbert space appears. 

The following result illustrates the connection between thermalization
and (shell-)ergodicity.
\begin{prop}
\label{thermal-ergo}An observable $A$ thermalizes if and only if
it is shell-ergodic.
\end{prop}
\begin{proof}
The $\Leftarrow$ direction is obvious. For the other direction, thermalization
means $\tr[(\overline{A(t)}-\langle A\rangle_{V})\rho_{0}]=0$, $\forall\rho_{0}\in\mathcal{S}_{V}$.
By Lemma \ref{lem:nondegenerate} (see Appendix A), $\Pi_{V}(\overline{A(t)}-\langle A\rangle_{V})\Pi_{V}=0$.
Noting that $\overline{A(t)}$ commutes with $\Pi_{V}$, we obtain
$(\overline{A(t)}-\langle A\rangle_{V})\Pi_{V}=0$.
\end{proof}
At this point we are ready to recall the definition of ETH. There
are two points to note. First, the ETH is never supposed to be valid
for the entire spectrum but only for the levels in some shell, here
$V$. The other point is that ETH is naturally made up of two statements,
a diagonal one and an off-diagonal one. It will be useful to separate
them.
\begin{defn}
ETH-D. An observable $A$ satisfies the diagonal eigenstate thermalization
hypothesis with respect to $V$ and $\rho_{V}$, if 
\begin{equation}
\Pi_{n}A\Pi_{n}=\langle A\rangle_{V}\Pi_{n},\quad\forall E_{n}\in V.\label{eq:ETHD}
\end{equation}

Of course if the eigen-projectors $\Pi_{n}$ are one-dimensional this
reduces to the standard, diagonal part of the ETH given in many references. 
\end{defn}
\begin{defn}
ETH-O. An observable $A$ satisfies the off-diagonal eigenstate thermalization
hypothesis with respect to $V$ and $\rho_{V}$, if $\Pi_{n}A\Pi_{m}=0,$
$\forall E_{n},E_{m}\in V,n\neq m$. 

If an observable satisfies both ETH-D and ETH-O we will simply talk
of ETH. ETH-O can have an important impact on the relaxation time
to the equilibrium state but not on the nature of the equilibrium
state itself \footnote{The reader may have noticed that we use the infinite-time average
in order to define quantities that have approached equilibrium. This
is almost a mathematical necessity in order to define equilibrium
quantities unambiguously though somewhat controversial from the physical
point of view. A discussion of this issue is outside the scope of
this paper. }. 

We now come to one of the main results of this paper. 
\end{defn}
\begin{prop}
An observable $A$ is shell-ergodic if and only if it satisfies ETH-D.
\end{prop}
\begin{proof}
It is possible to give a proof of this fact using Definition \ref{V-ergodic}.
However using the characterization of Proposition \ref{Ergo-alt}
the proof is particularly elementary. Consider first the (standard)
$\Leftarrow$ direction. Simply sum Eq.~(\ref{eq:ETHD}) for all
$n$ such that $E_{n}\in V$. Using the fact that $\overline{A(t)}=\sum_{n}\Pi_{n}A\Pi_{n}$
we obtain
\begin{equation}
\sum_{E_{k}\in V}\Pi_{k}A\Pi_{k}=\overline{A(t)}\Pi_{V}=\langle A\rangle_{V}\Pi_{V},
\end{equation}
that is, shell-ergodicity. The proof of the other implication is equally
trivial. Simply multiply both sides by $\Pi_{n}$ with $E_{n}\in V$
and we obtain Eq.~(\ref{eq:ETHD}).
\end{proof}
Recalling Proposition \ref{thermal-ergo}, we obtain the following.
\begin{cor*}
\emph{An observable $A$ thermalizes if and only if it satisfies ETH-D. }
\end{cor*}
The standard notion of ergodicity, however, is a property of a dynamical
systems, and not of single observables. This is particularly evident
in the characterization 1 of Theorem \ref{Characterisations-of-ergodicity}.
We call this property metric indecomposability to avoid confusion.
Roughly speaking, shell-ergodicity for a sufficiently large class
of observables should become equivalent to metric indecomposability.
Moreover, in the classical setting, it is usually believed (see e.g.~\citep{landau_statistical_1980})
that metric indecomposability implies that the only equilibrium state
(i.e., invariant measure) is the microcanonical one \footnote{Truthfully, the correct statement is that ergodicity (at fixed energy)
implies that the only invariant \emph{absolutely continuous }measure
is the microcanonical one \citep{cornfeld_ergodic_1982}.}.

In any case we will give a characterization of metric indecomposability
as it arises in quantum mechanics.

Let $\mathcal{E}_{t}^{\ast}$ (the star indicates Hilbert--Schmidt
adjoint) denote the Heisenberg evolution operator. It is easy to see
that shell-ergodicity means that 
\begin{equation}
\mathcal{T}(X)\equiv\overline{\mathcal{E}_{t}^{\ast}}(\Pi_{V}X\Pi_{V})=\Pi_{V}\langle X\rangle_{V},\label{eq:ergo_HV}
\end{equation}
where $\mathcal{T}:\mathcal{B}(\mathcal{H}_{V})\longrightarrow\mathcal{B}(\mathcal{H}_{V})$
is the restriction of $\overline{\mathcal{E}_{t}^{\ast}}$ to $\mathcal{B}(\mathcal{H}_{V})$.
If the above were true for all $X\in\mathcal{B}(\mathcal{H})$
then $\mathcal{T}$ would be equal to $\mathcal{T}_{MC}\equiv|\Pi_{V}\rangle_{HS}\langle\rho_{V}|_{HS}$
where we used Hilbert--Schmidt (HS) notation. Incidentally, this
is the definition of ergodicity in the theory of quantum semi-groups.
Moreover, this would imply at once that the invariant state $\rho_{V}$
is the microcanonical one. In fact the superoperator $\overline{\mathcal{E}_{t}^{\ast}}$
as well as $\mathcal{T}$ are HS self-adjoint (this is a statement
of the von Neumann ergodic theorem) and $\mathcal{T}=|\Pi_{V}\rangle\langle\rho_{V}|=\mathcal{T^{\ast}}$
implies $\rho_{V}=\Pi_{V}/\tr\Pi_{V}\equiv\rho_{MC}$. 

This appealing possibility, however, can never be realized in finite
dimension. In fact in this case one can explicitly compute
\begin{equation}
\overline{A}\Pi_{V}=\sum_{E_{n}\in V}\Pi_{n}A\Pi_{n},
\end{equation}
and this expression cannot be proportional to $\Pi_{V}$ for all $A$
unless $\Pi_{V}$ is one-dimensional\footnote{Simply take $A$ diagonal in $\mathcal{B}(\mathcal{H}_{V})$
but not proportional to $\Pi_{V}$. Then $\overline{A}\Pi_{V}=A$
which cannot be $\propto\Pi_{V}$. }, in which case we must have $\rho_{V}=|E_{n}\rangle\langle E_{n}|$
for some $E_{n}$.

The conclusion of this elementary argument is the following. In quantum
mechanics all the eigenspaces are invariant subspaces. Each eigenspace
with$H=E_{n}$ is trivially metrically indecomposable only if it is
one-dimensional. We will return to this point.

If $\Pi_{V}$ is not one-dimensional then Eq.~(\ref{eq:ergo_HV})
cannot be true for all $X$. Hence we are led to consider an approximate
version of Eq.~(\ref{eq:ergo_HV}), such as $\Vert (\mathcal{T-}\mathcal{T}_{MC})(X)\Vert \le\epsilon\Vert X\Vert $
(the norm used is the operator norm). For the same reasons as before,
even this relaxed form cannot be satisfied for all observables $X$.
In other words, the best we can hope for is \emph{approximate} shell-ergodicity
valid for \emph{some} observables. Given this state of affairs we
must modify our definitions in order to take into account possible
deviations from the ideal case. 

\begin{customdef}{1\ensuremath{'}} We say that $A$ thermalizes to
precision $\epsilon$ (or $\epsilon$-thermalizes) if $|\overline{\tr(A(t)\rho_{0})}-\langle A\rangle_{V}|\le\epsilon\Vert A\Vert $
for all $\rho_{0}\in\mathcal{S}_{V}$.

\end{customdef}

\begin{customdef}{2a\ensuremath{'}} We say that an observable $A$
is strongly $\epsilon$-shell-ergodic if $\Vert(\overline{A(t)}-\langle A\rangle_{V})\Pi_{V}\Vert\le\epsilon\Vert A\Vert$.

\end{customdef}

\begin{customdef}{2b\ensuremath{'}} We say that an observable $A$
is $\epsilon$-shell-ergodic if $|\overline{\langle A(t)A\rangle_{V}}-(\langle A\rangle_{V})^{2}|\le\epsilon^{2}\Vert A\Vert^{2}$.

\end{customdef}

\begin{customdef}{3\ensuremath{'}}We say that an observable $A$
satisfies ETH-D to precision $\epsilon$ if $\Vert \Pi_{n}A\Pi_{n}-\langle A\rangle_{V}\Pi_{n}\Vert \le\epsilon\Vert A\Vert$
for all $E_{n}\in V$.  

\end{customdef}

Note that for a non-degenerate spectrum Definition 3$'$ reduces to
$|\langle n|A|n\rangle-\langle A\rangle_{V}|\le\epsilon\Vert A\Vert$
for all $E_{n}\in V$.

As we have seen, shell-ergodicity for a single observable does not
imply that the invariant state is the microcanonical one. However
it implies that all possible ensemble averages are equal. Indeed the
same holds true for the $\epsilon$ version. Assume $A$ thermalizes
to precision $\epsilon$. Then we have $|\tr(A\rho_{V}')-\langle A\rangle_{V}|\le\epsilon\Vert A\Vert $
where $\rho_{V}'=\overline{\mathcal{E}_{t}}(\rho_{0})$ is any possible
invariant state. In other words we can always assume that the invariant
state is the one we prefer, e.g.~the microcanonical one, and have
thermalization to precision at most $2\epsilon$. 

It is straightforward to see that definitions 2a$'$ and 3$'$ are
equivalent since $\Vert(\overline{A(t)}-\langle A\rangle_{V})\Pi_{V}\Vert=\max_{E_{n}\in V}\Vert\Pi_{n}A\Pi_{n}-\langle A\rangle_{V}\Pi_{n}\Vert$.
Moreover we have (see Appendix B):
\begin{prop}
 The notions of thermalization, strong shell-ergodicity, and ETH-D
all coincide with the same $\epsilon$.
\end{prop}
For completeness we also relate Definitions 2a\emph{$'$} and 2b\emph{$'$}:
\begin{prop}
$A$ strongly $\epsilon$-shell-ergodic $\Rightarrow$ $A$ $\epsilon$-shell-ergodic,
whereas $A$ $\epsilon$-shell-ergodic $\Rightarrow$ $A$ is strongly
$\sqrt{\Vert \rho_{V}^{-1}\Pi_{V}\Vert} \epsilon$-shell-ergodic.
\end{prop}
Given the relaxed version of ETH-D one may think that it suffices
to satisfy $|\langle A\rangle_{E_{n}}-\langle A\rangle_{V}|\le\epsilon\Vert A\Vert $
for \emph{most $E_{n}\in V$. }However, it is easy to see that, even
a single deviation, say at $E_{n_{0}}$, implies that one cannot have
thermalization for all initial states. Thermalization is again restored
if one restricts to all initial states $\rho_{0}\in\mathcal{S}_{V}$
such that $\Pi_{n_{0}}\rho_{0}\Pi_{n_{0}}=0$. This remark is useful
given that in many cases one can prove a weak version of ETH-D where
the fraction of states in $V$ not satisfying ETH-D is suitably small
\citep{biroli_effect_2010,mori_weak_2016}. This is essentially also
the \emph{weak} ergodicity breaking discussed in \citep{turner_weak_2018}
and it implies thermalization for all the initial states except for
a small fraction. In the mathematical literature this marks the departure
from \emph{quantum unique ergodicity }to simple\emph{ quantum ergodicity. }

Although outside the scope of this work, here we will briefly comment
on conditions which guarantee thermalization. An almost trivial condition
is that the observable in question be proportional to the identity
in $\mathcal{H}_{V}$. For a Hamiltonian with non-degenerate spectrum
in $V$, this condition is in fact \emph{equivalent} to the full ETH.
For a local observable of the form $O=N^{-1}\sum_{x}O_{x}$, outside
of critical points the (quantum) central limit theorem essentially
implies that $(O-\langle O\rangle_{V})\Pi_{V}\simeq0$
for large system sizes $N$. Using this approach weak ETH, and, in
turn, thermalization for almost all initial states, was proved in
\citep{mori_weak_2016} for such observables.

\paragraph*{Origin of ETH and its classical version.}

Let us now return to classical dynamical systems and see how classical
statistical mechanics is built. There are two possible approaches. 

a) One possibility is to fix the Hamiltonian to have energy $E$.
Despite the fact that achieving perfect isolation is nearly impossible
(gravitational fields and cosmic rays are not easily screened out),
at least in principle this is perfectly legitimate in classical mechanics.
This defines $M_{E}=\{ x|H(x)=E\} $ as the invariant manifold.
Let $\langle f\rangle_{E}=\int_{M_{E}}\!f(x)\,\mathrm{d}\mu(x)$ denote
the phase space average of the observable $f$ on $M_{E}$. If $M_{E}$
were metrically indecomposable we would have $\overline{f}(x)=\langle f\rangle_{E}$
for almost all $x$ in $M_{E}$ for all essentially bounded observable
functions. In this setting the entropy is defined as $S_{S}(E)=k\ln\omega(E)$
with $\omega(E)=\int\!\mathrm{d}x\,\delta(H(x)-E)$, and one can go
on and define all the thermodynamic functions.  However this is not
the approach that is usually taken to build statistical mechanics. 

b) Indeed often one does not stop here but rather considers a thickened
shell $V_{\overline{E},\Delta}=\{x\in\Gamma|\overline{E}\le H(x)\le\overline{E}+\Delta\}$.
There seems to be no ontological reason to consider this setting,
as classically the energy of a truly isolated system is fixed to infinite
precision. The standard motivation is that this picture applies to
``almost isolated systems.'' The phase space averages over $M_{E}$
are replaced with uniform averages over $V_{\overline{E},\Delta}$.
In order for this to make sense one must require that for almost any
$E\in I=[\overline{E},\overline{E}+\Delta]$ the system is metrically
indecomposable. Note that the quantum equivalent of this situation
corresponds to a Hamiltonian whose spectrum $E_{n}$ in $[\overline{E},\overline{E}+\Delta]$
is non-degenerate. In quantum mechanics this is certainly a very common
property. For a given observable function $f$ (defined on $V_{\overline{E},\Delta}$)
we say it is $M$-ergodic if $\overline{f}(x)=\langle f\rangle_{E}$
for almost all $x\in M_{E}$ and for almost all $E\in I$. In this
setting phase space averages are defined as
\[
\langle f\rangle_{V_{\overline{E},\Delta}}\equiv\frac{\int_{V_{\overline{E},\Delta}}\!\mathrm{d}x\,f(x)}{\int_{V_{\overline{E},\Delta}}\!\mathrm{d}x}=\frac{\int_{I}\mathrm{d}E\,\omega(E)\langle f\rangle_{E}}{\int_{I}\mathrm{d}E\,\omega(E)},
\]
whereas the entropy is defined as $S_{V}(\overline{E})=\ln\Omega=\ln\int_{\overline{E}}^{\overline{E}+\Delta}\mathrm{d}E\,\omega(E)$. 

Of course, we want averages computed with approach a) to be equal
to those computed with b). Hence we require
\begin{equation}
\langle f\rangle_{E}=\langle f\rangle_{V_{\overline{E},\Delta}}\label{eq:ETH-C}
\end{equation}
 for almost all $E\in I$. This is clearly reminiscent of ETH-D. Truthfully,
this is equivalent to ETH-D only if $\Pi_{n}$ are one-dimensional
\footnote{Consider a) \unexpanded{$\langle A\rangle_{E_{n}}=\langle A\rangle_{V}$}
and b) \unexpanded{$\Pi_{n}A\Pi_{n}=\langle A\rangle_{V}\Pi_{n}$}.
Clearly b) $\Rightarrow$ a) but a) $\nRightarrow$ b) unless the
$\Pi_{n}$ are one-dimensional. }, so we call it ETH-C (classical). Its quantum mechanical version
is $\langle A\rangle_{E_{n}}=\langle A\rangle_{V}$.

Obviously if $f$ is $M$-ergodic and satisfies ETH-C then $f$ is
shell-ergodic, namely, $\overline{f}(x)=\langle f\rangle_{V_{\overline{E},\Delta}}$
for $V$-almost any $x\in V_{\overline{E},\Delta}$, which is what
we wanted. We see a clear parallel with the quantum world. It is also
clear that Eq.~(\ref{eq:ETH-C}), as well as shell-ergodicity, cannot
be satisfied for all functions $f$ (simply take an $f$ which is
not constant over different $M_{E}$) and in general can be valid
only approximately. It is the introduction of the shell $V_{\overline{E},\Delta}$
that forces us to consider approximate thermalization or ergodicity. 

The equivalence between approach a) and approach b) is usually not
discussed at length. A necessary condition is that $\langle f\rangle_{E}$
is a smooth function of $E$ and $\Delta$ sufficiently small. Assuming
$M_{\overline{E}}$ is sufficiently well behaved (a Lipschitz domain)
and $\langle f\rangle_{E}$ is differentiable as a function of $E$,
we have, as $\Delta\to0$, 
\begin{equation}
\langle f\rangle_{V_{\overline{E},\Delta}}=\langle f\rangle_{\overline{E}}+\frac{\Delta}{2}\langle f\rangle'_{\overline{E}}+O(\Delta^{2}).
\end{equation}
From the above, an estimate for the relative error is 
\begin{equation}
\left|\frac{\langle f\rangle_{V_{\overline{E},\Delta}}-\langle f\rangle_{\overline{E}}}{\langle f\rangle_{\overline{E}}}\right|\lesssim\frac{\Delta}{2\epsilon_{f}},
\end{equation}
where the energy scale $\epsilon_{f}$ is $\epsilon_{f}=\langle f\rangle_{\overline{E}}/|\langle f\rangle'_{\overline{E}}|.$ 

All in all, in order for $f$ to thermalize, we need metric indecomposability
for almost all $E\in I$, \emph{and} ETH-C, Eq.~(\ref{eq:ETH-C}),
for which a convenient proxy is given by $\Delta\ll\epsilon_{f}$.
For the Hamiltonian function $\epsilon_{H}=\overline{E}$ and we obtain
the standard requirement $\Delta/\overline{E}\ll1$. 

Let us now go back to the quantum realm. Now possibility a) is not
allowed for at least two reasons. First we can argue (as in \citep{landau_statistical_1980})
that the uncertainty in energy is a consequence of the system not
being exactly isolated. In this case one cannot be in an exact eigenstate
because of a time-energy uncertainty where $\Delta t$ is the duration
of the interaction process. Likewise, interactions with an environment
would cause a broadening of levels. These arguments do not apply to
a truly isolated system. For a truly isolated system however, we can
say that we could not define a meaningful entropy function. So considering
scenario b) becomes a necessity in quantum mechanics. 

Reproducing the classical argument, we then need metric indecomposability
for all the levels in a certain shell (now called $V$). As we have
seen, in quantum mechanics, this is simply the requirement that the
levels in $V$ are non-degenerate, a quite common property. After
that we still demand equality of the two scenarios, i.e., $\langle A\rangle_{E_{n}}=\langle A\rangle_{V}$
which, as we have seen, is equivalent to thermalization.

Note that any invariant state can be written as $\rho_{V}=\sum_{n}p_{n}\Pi_{n}/\tr\Pi_{n}$.
Then the phase space average can be written as
\[
\langle A\rangle_{V}=\sum_{E_{n}\in I}p_{n}\langle A\rangle_{E_{n}}=\frac{\int_{I}\mathrm{d}E\,\tilde{\omega}(E)\langle A\rangle_{E}}{\int_{I}\mathrm{d}E\,\tilde{\omega}(E)}
\]
where we introduced the density of levels function 
\[
\tilde{\omega}(E)\equiv\sum_{E_{n}\in I}\delta(E-E_{n})p_{n}
\]
 ---in other words it is formally precisely the classical one, upon
identifying $\tilde{\omega}=\omega$. 

\paragraph*{Conclusions. }

We show that a proper definition of ergodicity in the quantum framework
is equivalent to the standard notion of thermalization for all initial
states generally used in the literature. Moreover we prove that ergodicity
is equivalent to the diagonal part of the eigenstate thermalization
hypothesis (ETH) implying equivalence between ETH and thermalization,
thus resolving a current conjecture. The arguments are elementary
and in fact similar results also hold in the classical framework once
suitably translated to the corresponding language. Indeed, we show
that ETH is also present and implicitly assumed in the foundations
of classical statistical mechanics. One point where the analogy breaks
down is the conceptual impossibility in quantum mechanics of fixing
the energy of an isolated statistical system to infinite precision. 
\begin{acknowledgments}
This work was partially supported by the Air Force Research Laboratory
award no.~FA8750-18-1-0041 and (partially) by the Office of the Director
of National Intelligence (ODNI), Intelligence Advanced Research Projects
Activity (IARPA), via the U.S.~Army Research Office contract W911NF-17-C-0050.
The views and conclusions contained herein are those of the authors
and should not be interpreted as necessarily representing the official
policies or endorsements, either expressed or implied, of the ODNI,
IARPA, or the U.S.~Government. The U.S.~Government is authorized
to reproduce and distribute reprints for Governmental purposes notwithstanding
any copyright annotation thereon.
\end{acknowledgments}

\bibliographystyle{apsrev4-1}
\bibliography{ergodicity_ETH}

\begin{thebibliography}{23}%
\makeatletter
\providecommand \@ifxundefined [1]{%
 \@ifx{#1\undefined}
}%
\providecommand \@ifnum [1]{%
 \ifnum #1\expandafter \@firstoftwo
 \else \expandafter \@secondoftwo
 \fi
}%
\providecommand \@ifx [1]{%
 \ifx #1\expandafter \@firstoftwo
 \else \expandafter \@secondoftwo
 \fi
}%
\providecommand \natexlab [1]{#1}%
\providecommand \enquote  [1]{``#1''}%
\providecommand \bibnamefont  [1]{#1}%
\providecommand \bibfnamefont [1]{#1}%
\providecommand \citenamefont [1]{#1}%
\providecommand \href@noop [0]{\@secondoftwo}%
\providecommand \href [0]{\begingroup \@sanitize@url \@href}%
\providecommand \@href[1]{\@@startlink{#1}\@@href}%
\providecommand \@@href[1]{\endgroup#1\@@endlink}%
\providecommand \@sanitize@url [0]{\catcode `\\12\catcode `\$12\catcode
  `\&12\catcode `\#12\catcode `\^12\catcode `\_12\catcode `\%12\relax}%
\providecommand \@@startlink[1]{}%
\providecommand \@@endlink[0]{}%
\providecommand \url  [0]{\begingroup\@sanitize@url \@url }%
\providecommand \@url [1]{\endgroup\@href {#1}{\urlprefix }}%
\providecommand \urlprefix  [0]{URL }%
\providecommand \Eprint [0]{\href }%
\providecommand \doibase [0]{http://dx.doi.org/}%
\providecommand \selectlanguage [0]{\@gobble}%
\providecommand \bibinfo  [0]{\@secondoftwo}%
\providecommand \bibfield  [0]{\@secondoftwo}%
\providecommand \translation [1]{[#1]}%
\providecommand \BibitemOpen [0]{}%
\providecommand \bibitemStop [0]{}%
\providecommand \bibitemNoStop [0]{.\EOS\space}%
\providecommand \EOS [0]{\spacefactor3000\relax}%
\providecommand \BibitemShut  [1]{\csname bibitem#1\endcsname}%
\let\auto@bib@innerbib\@empty
\bibitem [{\citenamefont {Pal}\ and\ \citenamefont
  {Huse}(2010)}]{pal_many-body_2010}%
  \BibitemOpen
  \bibfield  {author} {\bibinfo {author} {\bibfnamefont {A.}~\bibnamefont
  {Pal}}\ and\ \bibinfo {author} {\bibfnamefont {D.~A.}\ \bibnamefont {Huse}},\
  }\href {\doibase 10.1103/PhysRevB.82.174411} {\bibfield  {journal} {\bibinfo
  {journal} {Phys. Rev. B}\ }\textbf {\bibinfo {volume} {82}},\ \bibinfo
  {pages} {174411} (\bibinfo {year} {2010})}\BibitemShut {NoStop}%
\bibitem [{\citenamefont {Abanin}\ \emph {et~al.}(2018)\citenamefont {Abanin},
  \citenamefont {Altman}, \citenamefont {Bloch},\ and\ \citenamefont
  {Serbyn}}]{abanin_ergodicity_2018}%
  \BibitemOpen
  \bibfield  {author} {\bibinfo {author} {\bibfnamefont {D.~A.}\ \bibnamefont
  {Abanin}}, \bibinfo {author} {\bibfnamefont {E.}~\bibnamefont {Altman}},
  \bibinfo {author} {\bibfnamefont {I.}~\bibnamefont {Bloch}}, \ and\ \bibinfo
  {author} {\bibfnamefont {M.}~\bibnamefont {Serbyn}},\ }\href
  {http://arxiv.org/abs/1804.11065} {\bibfield  {journal} {\bibinfo  {journal}
  {arXiv:1804.11065 [cond-mat, physics:quant-ph]}\ } (\bibinfo {year}
  {2018})},\ \bibinfo {note} {arXiv: 1804.11065}\BibitemShut {NoStop}%
\bibitem [{\citenamefont {Deutsch}(1991)}]{deutsch_quantum_1991}%
  \BibitemOpen
  \bibfield  {author} {\bibinfo {author} {\bibfnamefont {J.~M.}\ \bibnamefont
  {Deutsch}},\ }\href {\doibase 10.1103/PhysRevA.43.2046} {\bibfield  {journal}
  {\bibinfo  {journal} {Phys. Rev. A}\ }\textbf {\bibinfo {volume} {43}},\
  \bibinfo {pages} {2046} (\bibinfo {year} {1991})}\BibitemShut {NoStop}%
\bibitem [{\citenamefont {Srednicki}(1994)}]{srednicki_chaos_1994}%
  \BibitemOpen
  \bibfield  {author} {\bibinfo {author} {\bibfnamefont {M.}~\bibnamefont
  {Srednicki}},\ }\href {\doibase 10.1103/PhysRevE.50.888} {\bibfield
  {journal} {\bibinfo  {journal} {Phys. Rev. E}\ }\textbf {\bibinfo {volume}
  {50}},\ \bibinfo {pages} {888} (\bibinfo {year} {1994})}\BibitemShut
  {NoStop}%
\bibitem [{\citenamefont {Rigol}\ \emph {et~al.}(2007)\citenamefont {Rigol},
  \citenamefont {Dunjko}, \citenamefont {Yurovsky},\ and\ \citenamefont
  {Olshanii}}]{rigol_relaxation_2007}%
  \BibitemOpen
  \bibfield  {author} {\bibinfo {author} {\bibfnamefont {M.}~\bibnamefont
  {Rigol}}, \bibinfo {author} {\bibfnamefont {V.}~\bibnamefont {Dunjko}},
  \bibinfo {author} {\bibfnamefont {V.}~\bibnamefont {Yurovsky}}, \ and\
  \bibinfo {author} {\bibfnamefont {M.}~\bibnamefont {Olshanii}},\ }\href
  {\doibase 10.1103/PhysRevLett.98.050405} {\bibfield  {journal} {\bibinfo
  {journal} {Phys. Rev. Lett.}\ }\textbf {\bibinfo {volume} {98}},\ \bibinfo
  {pages} {050405} (\bibinfo {year} {2007})}\BibitemShut {NoStop}%
\bibitem [{\citenamefont {Gogolin}\ and\ \citenamefont
  {Eisert}(2016)}]{gogolin_equilibration_2016}%
  \BibitemOpen
  \bibfield  {author} {\bibinfo {author} {\bibfnamefont {C.}~\bibnamefont
  {Gogolin}}\ and\ \bibinfo {author} {\bibfnamefont {J.}~\bibnamefont
  {Eisert}},\ }\href {\doibase 10.1088/0034-4885/79/5/056001} {\bibfield
  {journal} {\bibinfo  {journal} {Rep. Prog. Phys.}\ }\textbf {\bibinfo
  {volume} {79}},\ \bibinfo {pages} {056001} (\bibinfo {year}
  {2016})}\BibitemShut {NoStop}%
\bibitem [{Note1()}]{Note1}%
  \BibitemOpen
  \bibinfo {note} {Additionally, in the mathematical literature, the related
  concept of \protect \emph {quantum (unique) ergodicity}, which we might call
  ETH at infinite energy, has been proved rigorously for single particle
  Hamiltonians whose classical counterpart is ergodic (see e.g.~\protect \citep
  {zelditch_quantum_1990,sunada_quantum_1997}).}\BibitemShut {Stop}%
\bibitem [{\citenamefont {Mazur}(1969)}]{mazur_non-ergodicity_1969}%
  \BibitemOpen
  \bibfield  {author} {\bibinfo {author} {\bibfnamefont {P.}~\bibnamefont
  {Mazur}},\ }\href {\doibase 10.1016/0031-8914(69)90185-2} {\bibfield
  {journal} {\bibinfo  {journal} {Physica}\ }\textbf {\bibinfo {volume} {43}},\
  \bibinfo {pages} {533} (\bibinfo {year} {1969})}\BibitemShut {NoStop}%
\bibitem [{\citenamefont {Suzuki}(1971)}]{suzuki_ergodicity_1971}%
  \BibitemOpen
  \bibfield  {author} {\bibinfo {author} {\bibfnamefont {M.}~\bibnamefont
  {Suzuki}},\ }\href {\doibase 10.1016/0031-8914(71)90226-6} {\bibfield
  {journal} {\bibinfo  {journal} {Physica}\ }\textbf {\bibinfo {volume} {51}},\
  \bibinfo {pages} {277} (\bibinfo {year} {1971})}\BibitemShut {NoStop}%
\bibitem [{\citenamefont {Ruelle}(1999)}]{ruelle_statistical_1999}%
  \BibitemOpen
  \bibfield  {author} {\bibinfo {author} {\bibfnamefont {D.}~\bibnamefont
  {Ruelle}},\ }\href@noop {} {{\selectlanguage {en}\emph {\bibinfo {title}
  {Statistical {Mechanics}: {Rigorous} {Results}}}}}\ (\bibinfo  {publisher}
  {World Scientific},\ \bibinfo {year} {1999})\BibitemShut {NoStop}%
\bibitem [{\citenamefont {Alicki}\ and\ \citenamefont
  {Fannes}(2001)}]{alicki_quantum_2001}%
  \BibitemOpen
  \bibfield  {author} {\bibinfo {author} {\bibfnamefont {R.}~\bibnamefont
  {Alicki}}\ and\ \bibinfo {author} {\bibfnamefont {M.}~\bibnamefont
  {Fannes}},\ }\href
  {http://www.oxfordscholarship.com/view/10.1093/acprof:oso/9780198504009.001.0001/acprof-9780198504009}
  {\emph {\bibinfo {title} {Quantum {Dynamical} {Systems}}}}\ (\bibinfo
  {publisher} {Oxford University Press},\ \bibinfo {year} {2001})\BibitemShut
  {NoStop}%
\bibitem [{\citenamefont {De~Palma}\ \emph {et~al.}(2015)\citenamefont
  {De~Palma}, \citenamefont {Serafini}, \citenamefont {Giovannetti},\ and\
  \citenamefont {Cramer}}]{de_palma_necessity_2015}%
  \BibitemOpen
  \bibfield  {author} {\bibinfo {author} {\bibfnamefont {G.}~\bibnamefont
  {De~Palma}}, \bibinfo {author} {\bibfnamefont {A.}~\bibnamefont {Serafini}},
  \bibinfo {author} {\bibfnamefont {V.}~\bibnamefont {Giovannetti}}, \ and\
  \bibinfo {author} {\bibfnamefont {M.}~\bibnamefont {Cramer}},\ }\href
  {\doibase 10.1103/PhysRevLett.115.220401} {\bibfield  {journal} {\bibinfo
  {journal} {Phys. Rev. Lett.}\ }\textbf {\bibinfo {volume} {115}},\ \bibinfo
  {pages} {220401} (\bibinfo {year} {2015})}\BibitemShut {NoStop}%
\bibitem [{\citenamefont {Cornfeld}\ \emph {et~al.}(1982)\citenamefont
  {Cornfeld}, \citenamefont {Fomin},\ and\ \citenamefont
  {Sinai}}]{cornfeld_ergodic_1982}%
  \BibitemOpen
  \bibfield  {author} {\bibinfo {author} {\bibfnamefont {I.~P.}\ \bibnamefont
  {Cornfeld}}, \bibinfo {author} {\bibfnamefont {S.~V.}\ \bibnamefont {Fomin}},
  \ and\ \bibinfo {author} {\bibfnamefont {Y.~G.}\ \bibnamefont {Sinai}},\
  }\href {//www.springer.com/us/book/9781461569299} {\emph {\bibinfo {title}
  {Ergodic {Theory}}}},\ Grundlehren der mathematischen {Wissenschaften}\
  (\bibinfo  {publisher} {Springer-Verlag},\ \bibinfo {address} {New York},\
  \bibinfo {year} {1982})\BibitemShut {NoStop}%
\bibitem [{Note2()}]{Note2}%
  \BibitemOpen
  \bibinfo {note} {The reader may have noticed that we use the infinite-time
  average in order to define quantities that have approached equilibrium. This
  is almost a mathematical necessity in order to define equilibrium quantities
  unambiguously though somewhat controversial from the physical point of view.
  A discussion of this issue is outside the scope of this paper.}\BibitemShut
  {Stop}%
\bibitem [{\citenamefont {Landau}\ and\ \citenamefont
  {Lifshitz}(1980)}]{landau_statistical_1980}%
  \BibitemOpen
  \bibfield  {author} {\bibinfo {author} {\bibfnamefont {L.~D.}\ \bibnamefont
  {Landau}}\ and\ \bibinfo {author} {\bibfnamefont {E.~M.}\ \bibnamefont
  {Lifshitz}},\ }\href@noop {} {{\selectlanguage {English}\emph {\bibinfo
  {title} {Statistical {Physics}, {Third} {Edition}, {Part} 1: {Volume} 5}}}},\
  \bibinfo {edition} {3rd}\ ed.\ (\bibinfo  {publisher}
  {Butterworth-Heinemann},\ \bibinfo {address} {Amsterdam u.a},\ \bibinfo
  {year} {1980})\BibitemShut {NoStop}%
\bibitem [{Note3()}]{Note3}%
  \BibitemOpen
  \bibinfo {note} {Truthfully, the correct statement is that ergodicity (at
  fixed energy) implies that the only invariant \protect \emph {absolutely
  continuous }measure is the microcanonical one \protect \citep
  {cornfeld_ergodic_1982}.}\BibitemShut {Stop}%
\bibitem [{Note4()}]{Note4}%
  \BibitemOpen
  \bibinfo {note} {Simply take $A$ diagonal in $\protect \mathcal {B}\left
  (\protect \mathcal {H}_{V}\right )$ but not proportional to $\Pi _{V}$. Then
  $\protect \overline {A}\Pi _{V}=A$ which cannot be $\propto \Pi
  _{V}$.}\BibitemShut {Stop}%
\bibitem [{\citenamefont {Biroli}\ \emph {et~al.}(2010)\citenamefont {Biroli},
  \citenamefont {Kollath},\ and\ \citenamefont
  {Läuchli}}]{biroli_effect_2010}%
  \BibitemOpen
  \bibfield  {author} {\bibinfo {author} {\bibfnamefont {G.}~\bibnamefont
  {Biroli}}, \bibinfo {author} {\bibfnamefont {C.}~\bibnamefont {Kollath}}, \
  and\ \bibinfo {author} {\bibfnamefont {A.~M.}\ \bibnamefont {Läuchli}},\
  }\href {\doibase 10.1103/PhysRevLett.105.250401} {\bibfield  {journal}
  {\bibinfo  {journal} {Phys. Rev. Lett.}\ }\textbf {\bibinfo {volume} {105}},\
  \bibinfo {pages} {250401} (\bibinfo {year} {2010})}\BibitemShut {NoStop}%
\bibitem [{\citenamefont {Mori}(2016)}]{mori_weak_2016}%
  \BibitemOpen
  \bibfield  {author} {\bibinfo {author} {\bibfnamefont {T.}~\bibnamefont
  {Mori}},\ }\href {http://arxiv.org/abs/1609.09776} {\bibfield  {journal}
  {\bibinfo  {journal} {arXiv:1609.09776 [cond-mat, physics:quant-ph]}\ }
  (\bibinfo {year} {2016})},\ \bibinfo {note} {arXiv: 1609.09776}\BibitemShut
  {NoStop}%
\bibitem [{\citenamefont {Turner}\ \emph {et~al.}(2018)\citenamefont {Turner},
  \citenamefont {Michailidis}, \citenamefont {Abanin}, \citenamefont {Serbyn},\
  and\ \citenamefont {Papić}}]{turner_weak_2018}%
  \BibitemOpen
  \bibfield  {author} {\bibinfo {author} {\bibfnamefont {C.~J.}\ \bibnamefont
  {Turner}}, \bibinfo {author} {\bibfnamefont {A.~A.}\ \bibnamefont
  {Michailidis}}, \bibinfo {author} {\bibfnamefont {D.~A.}\ \bibnamefont
  {Abanin}}, \bibinfo {author} {\bibfnamefont {M.}~\bibnamefont {Serbyn}}, \
  and\ \bibinfo {author} {\bibfnamefont {Z.}~\bibnamefont {Papić}},\ }\href
  {\doibase 10.1038/s41567-018-0137-5} {\bibfield  {journal} {\bibinfo
  {journal} {Nature Physics}\ }\textbf {\bibinfo {volume} {14}},\ \bibinfo
  {pages} {745} (\bibinfo {year} {2018})}\BibitemShut {NoStop}%
\bibitem [{Note5()}]{Note5}%
  \BibitemOpen
  \bibinfo {note} {Consider a) $\langle A\rangle _{E_{n}}=\langle A\rangle
  _{V}$ and b) $\Pi _{n}A\Pi _{n}=\langle A\rangle _{V}\Pi _{n}$. Clearly b)
  $\Rightarrow $ a) but a) $\nRightarrow $ b) unless the $\Pi _{n}$ are
  one-dimensional.}\BibitemShut {Stop}%
\bibitem [{\citenamefont {Zelditch}(1990)}]{zelditch_quantum_1990}%
  \BibitemOpen
  \bibfield  {author} {\bibinfo {author} {\bibfnamefont {S.}~\bibnamefont
  {Zelditch}},\ }\href {\doibase 10.1016/0022-1236(90)90021-C} {\bibfield
  {journal} {\bibinfo  {journal} {Journal of Functional Analysis}\ }\textbf
  {\bibinfo {volume} {94}},\ \bibinfo {pages} {415} (\bibinfo {year}
  {1990})}\BibitemShut {NoStop}%
\bibitem [{\citenamefont {Sunada}(1997)}]{sunada_quantum_1997}%
  \BibitemOpen
  \bibfield  {author} {\bibinfo {author} {\bibfnamefont {T.}~\bibnamefont
  {Sunada}},\ }in\ \href {\doibase 10.1007/978-3-0348-8938-4_10}
  {{\selectlanguage {en}\emph {\bibinfo {booktitle} {Progress in {Inverse}
  {Spectral} {Geometry}}}}}\ (\bibinfo  {publisher} {Birkhäuser, Basel},\
  \bibinfo {year} {1997})\ pp.\ \bibinfo {pages} {175--196}\BibitemShut
  {NoStop}%
\end{thebibliography}%

\appendix

\section{Proofs of the lemmas}
\begin{lem}
\label{lem:zero}If $\tr(X^{\dagger}X\rho)=0$ then $XP=0$ (and $PX^{\dagger}=0$)
where $P$ is the orthogonal projector onto the support of $\rho$.
\end{lem}
\begin{proof}
Because $\rho=P\sqrt{\rho}PP\sqrt{\rho}P$ (since $P$ and $\rho$ commute),
\[
\tr(X^{\dagger}X\rho)=\tr[(XP\sqrt{\rho}P)^{\dagger}(XP\sqrt{\rho}P)] =\Vert XP\sqrt{\rho}P\Vert _{HS}^{2}=0,
\]
 which implies $XP\sqrt{\rho}P=0.$ Inside the range of $P$, $\sqrt{\rho}$
is invertible and we can multiply by $P\sqrt{\rho}^{-1}P$ and get
$XP=0$. 
\end{proof}
\begin{lem}
\label{lem:nondegenerate}If $\tr(X\rho_{0})=0$, $\forall\rho_{0}\in\mathcal{S}_{V}$
then $\Pi_{V}X\Pi_{V}=0$ (in fact the two statements are equivalent).
\end{lem}
\begin{proof}
Although $\mathcal{S}_{V}$ is not a linear space it is possible to
find $n^{2}$ linearly independent matrices in it, where $n=\dim\mathcal{H}_{V}$.
The equation $\tr(X\rho_{0})=0$ can be written as $\langle X^{\dagger}|\rho_{0}\rangle_{HS}=0$
using the Hilbert--Schmidt scalar product. Let $|e_{i}\rangle$,
$i=1,2,\ldots,n$ be a basis of $\mathcal{H}_{V}$. Then the last
equation means that $|X\rangle_{HS}$ is perpendicular to $\mathrm{span}\{|e_{i}\rangle\langle e_{j}|\}$.
This means that $X=\Pi_{V}XQ+QX\Pi_{V}+QXQ$ with $Q=\1-\Pi_{V}$.
In other words $\Pi_{V}X\Pi_{V}=0$.
\end{proof}

\section{Proofs of the $\epsilon$-relations}
\begin{proof}[Proof of Proposition 4.]
The equivalence of ETH-D and strong shell-ergodicity to precision
$\epsilon$ was established in the main text. Here we complete the
proof by establishing an equivalence between thermalization and strong
shell-ergodicity.

Without loss of generality assume $\Vert A\Vert =1$. Consider
$X\equiv(\overline{A(t)}-\langle A\rangle_{V})\Pi_{V}=\Pi_{V}X\Pi_{V}$.
Note that $X$ is hermitian if $A$ is. The observable $A$ thermalizes
if $|\tr(X\rho_{0})|\le\epsilon$ for all $\rho_{0}\in\mathcal{S}_{V}$.
Since $\rho_{0}$ is a state we have $|\tr(X\rho_{0})|\le\Vert X\Vert $.
Hence $A$ strongly $\epsilon$-shell-ergodic implies $A$ $\epsilon$-thermalizes.
But we also have 
\begin{equation}
\Vert \Pi_{V}X\Pi_{V}\Vert =\sup_{\rho\in\mathcal{S}_{V}}|\tr(X\rho)|.
\end{equation}
 So $|\tr(X\rho)|\le\epsilon$ for all $\rho\in\mathcal{S}_{V}$
implies that the left hand side is also $\le\epsilon$, i.e., we have
the other direction with the same $\epsilon$.
\end{proof}
\begin{proof}[Proof of Proposition 5.]
Without loss of generality assume $\Vert A\Vert =1$.
First the $\Rightarrow$ direction. With the same $X$ defined previously
we have $X^{2}=[\overline{A}\,\overline{A}-2\overline{A}\langle A\rangle_{V}+(\langle A\rangle_{V})^{2}]\Pi_{V}$.
Hence $|\overline{\langle A(t)A\rangle_{V}}-(\langle A\rangle_{V})^{2}|=|\tr(X^{2}\rho_{V})|\le\Vert X^{2}\Vert =\Vert X\Vert ^{2}\le\epsilon^{2}$.
For the other direction note that for $Y$ a positive definite operator
($Y\ge0$) we have $\Vert \Pi_{V}Y\Pi_{V}\Vert \le|\tr(Y\rho_{V})|\Vert \rho_{V}^{-1}\Pi_{V}\Vert $.
In fact $\tr(Y\rho_{V})=\tr(\Pi_{V}Y\Pi_{V}\rho_{V})$
where $\Pi_{V}Y\Pi_{V}$ is again positive with spectral resolution
$\Pi_{V}Y\Pi_{V}=\sum_{n}\lambda_{n}|n\rangle\langle n|$. Hence $\tr(Y\rho_{V})=\sum_{n}\lambda_{n}\langle n|\rho_{V}|n\rangle$
with $\lambda_{n}\ge0$ and $\langle n|\rho_{V}|n\rangle\ge\min_{j}f_{j}$
where $f_{j}$ are the non-zero eigenvalues of $\rho_{V}$. Since
$1/\min_{j}f_{j}=\max_{j}1/f_{j}=\Vert \Pi_{V}\rho_{V}^{-1}\Pi_{V}\Vert $
we obtain 

\begin{equation}
\Vert \Pi_{V}Y\Pi_{V}\Vert \le\tr(\Pi_{V}Y\Pi_{V}\rho_{V})\Vert \rho_{V}^{-1}\Pi_{V}\Vert .
\end{equation}
 Now use this with $Y=X^{2}$ to obtain
\begin{align}
\Vert(\overline{A(t)}-\langle A\rangle_{V})\Pi_{V}\Vert^{2} & =\Vert \Pi_{V}X\Pi_{V}\Vert ^{2}\nonumber \\
 & =\Vert \Pi_{V}X^{2}\Pi_{V}\Vert \nonumber \\
 & \le|\overline{\langle A(t)A\rangle_{V}}-(\langle A\rangle_{V})^{2}|\Vert \rho_{V}^{-1}\Pi_{V}\Vert ,\label{eq:bound_X}
\end{align}
from which the result follows. 

\end{proof}

\end{document}